\documentclass[runningheads]{llncs}

\usepackage{llncsdoc}
\usepackage[utf8]{inputenc}
\usepackage{hyperref}
\usepackage{amsmath}
\usepackage{xcolor}
\newcommand{\TT}{\mathit{TT}}
\newcommand{\MT}{\mathit{MT}}
\newcommand{\RSA}{\mathit{RSA}}
\definecolor[named]{tucgreen}{RGB}{0,140,79}

\newcommand{\para}[1]{\smallskip\noindent\textbf{#1}}

\usepackage[colorinlistoftodos,prependcaption,textsize=tiny]{todonotes}

\title{\vspace{0cm} Risk-Limiting Tallies}
\author{Wojciech Jamroga$^{1,2}$,\ Peter B. Roenne$^{1}$,\ Peter Y. A. Ryan$^{1}$,\\ and Philip B. Stark$^{3}$}
\institute{
  Interdisciplinary Centre for Security, Reliability, and Trust, SnT, University of Luxembourg
  \and
  Institute of Computer Science, Polish Academy of Sciences, Warsaw, Poland
  \and
  Department of Statistics, University of California, Berkeley, US \\
  \email{\{wojciech.jamroga,peter.roenne,peter.ryan\}@uni.lu, stark@stat.berkeley.edu}
  }
\authorrunning{W. Jamroga \and P. B. Roenne \and P. Y. A. Ryan \and and P. B. Stark}
\titlerunning{Risk-Limiting Tallies}
\date{\today}

\begin{document}

\maketitle

\begin{abstract}
Many voter-verifiable, coercion-resistant schemes have been proposed, but even the most carefully designed systems necessarily leak information via the announced result. In corner cases, this may be problematic. For example, if all the votes go to one candidate then all vote privacy evaporates. The mere possibility of candidates getting no or few votes could have implications for security in practice: if a coercer demands that a voter cast a vote for such an unpopular candidate, then the voter may feel obliged to obey, even if she is confident that the voting system satisfies the standard coercion resistance definitions.
With complex ballots, there may also be a danger of ``Italian'' style (aka ``signature'') attacks: the coercer demands the voter cast a ballot with a specific, identifying pattern.

Here we propose an approach to tallying end-to-end verifiable schemes that avoids revealing all the votes but still achieves whatever confidence level in the announced result is desired. Now a coerced voter can claim that the required vote must be amongst those that remained shrouded. Our approach is based on the well-established notion of Risk-Limiting Audits, but here applied to the tally rather than to the audit. We show that this approach  counters coercion threats arising in extreme tallies and ``Italian'' attacks.
We illustrate our approach by applying it to the Selene scheme, and we extend the approach to Risk-Limiting Verification, where not all vote trackers are revealed, thereby enhancing the coercion mitigation properties of Selene.


\let\thempfn\relax
\footnotetext[0]{Licence: CC-BY-ND 2.0 \url{https://creativecommons.org/licenses/by-nd/2.0/}}

\end{abstract}


\section{Introduction}

Many verifiable voting schemes have been proposed that are designed to give a high level of resistance against coercion or vote buying~\cite{Benaloh94receipt,juels05coercion,crs:pav,selene,Ryan09prettygood}. However, it is typically assumed that little can be done about coercion threats in case of extreme outcomes, e.g., no or few votes for some candidates.
Unfortunately, such situations do happen in real elections. 
 If there is a (perceived) risk that candidate $X$ will get no votes, and the coercer tells the voter to vote for $X$, then the voter may feel obliged to comply even if the voting scheme satisfies the standard definitions of coercion resistance.
The possibility is more dangerous than it seems at the first glance. True, coercing for the low support candidate $X$ is unlikely to get him or her win. However, the coercer can use $X$ to construct what is effectively an abstention attack, and take away the votes from the main opponent of his preferred candidate. 
If the coercer prefers candidate $A$, he can help him win by coercing supporters of $B$ to vote for $X$.

Another difficulty that may arise is that of so-called ``Italian''-style attacks, also known as ``signature'' attacks: if the voting method allows for a large number of distinct ways of filling out the ballot, a coercer may require the voter to fill out the ballot with a distinctive pattern allowing it to be uniquely identified with high probability in the final tally.
This is especially an issue with long, complex ballots and with preferential voting schemes.
It can be countered by, for example, using homomorphic tallying techniques to compute the overall result without revealing the individual ballots, but this is computationally intensive and even then may leak some critical information \cite{teague}.

Here, we show that risk-limiting audit techniques~\cite{lindemanStark12} can be adapted to achieve whatever level of confidence in the outcome is required while ensuring that a proportion of the ballots remain shrouded. This allows us to significantly enhance the coercion resistance of verifiable schemes.
The Risk Limiting Tally (RLT) approach that we present here provides a simple way to guarantee voters plausible deniability against the above attacks: a coerced voter, who did not cast the ballot the way the coercer had demanded, can simply claim that the required ballot is amongst those unrevealed.

We also present a variant of the idea, \emph{Risk Limiting Verification} (RLV), where we ensure that a proportion of verification tokens remain unrevealed. The basic version of the Selene scheme \cite{selene} has the drawback that the the coercer can claim that the fake tracker provided by a voter is his own. We describe how RLV mitigates this.




A possible objection to RLT is that it is ``undemocractic'' not to count all votes.
However, the method allows the electoral outcome (i.e. the winner or winners) to be ascertained to
any desired level of statistical certainty.
Moreover, the sample of ballots will be drawn in such a way that every cast vote has an equal chance of being in the revealed sample, so there is no lack of fairness.
RLTs are related to
\emph{random sample voting} (RSV), due to Chaum~\cite{chaumRSV}, except that there the sample is drawn from the set of eligible voters, rather than from the cast votes.
If anything, RLTs seem to be more democratic in that in RSV voters who are not chosen might well feel excluded. Furthermore, in RLTs we are able to adjust the sample size after voting to achieve the desired confidence level. We note also that some tally algorithms, e.g. some forms of STV, intrinsically involve a probabilistic element.




The outline of this paper is as follows. In section \ref{sec:coercionverifiable} we discuss coercion-resistant and verifiable voting schemes in general, and Selene in particular. Section \ref{sec:RLA} briefly recalls Risk-Limiting Audits, and Section \ref{sec:RLT} introduces the techniques for RLTs. We present the actual protocol in Section \ref{sec:protocol}, and a brief security discussion in Section \ref{sec:security}. Section \ref{sec:rlv} discusses the risk-limiting verification. Related work is presented in Section \ref{sec:relatedwork}, and we conclude in Section \ref{sec:conclusions}.

\subsection{Contribution}

In this paper we present several contributions:

\begin{enumerate}
    \item The use of using risk-limiting techniques to shroud a proportion of votes, improving coercion resistance while achieving whatever confidence level is required.
    \item A novel extension of Risk-Limiting Audit (RLA) techniques to handle the situation in which we do not have an initial null hypothesis (reported outcome).
    \item A new test statistic for RLAs with operating characteristics that do not depend on the reported votes, only on the reported winner(s).
    \item Protocols to enable RLT for most end-to-end verifiable schemes, including strategies to ensure plausible deniability whatever the vote distribution is.
    \item Extension of the approach to Risk-Limiting Verification: the shrouding of a randomly selected subset of verification tokens to improve coercion resistance, in particular for the Selene scheme.
\end{enumerate}

\section{Coercion-Resistant and Verifiable Voting}\label{sec:coercionverifiable}

We set the scene by recalling the concepts of End-to-End Verifiability~\cite{Ryan15verifiability} and coercion-resistance~\cite{Benaloh94receipt,juels05coercion}, and showing an example scheme designed to balance the two requirements.

\subsection{An Outline of End-to-End Verifiable Voting}

The use of digital technologies to record and process votes might provide efficiency and convenience, but it can also bring serious new threats, in particular, virtually undetectable ways to manipulate votes on a large scale.
These concerns motivated the development of \emph{End-to-End Verifiable} (E2E~V) voting schemes.
Such schemes provide the voters with means to confirm that their vote is accurately included in the tally, without opening up possibilities of coercion or vote-buying.
This is usually accomplished by creating an encryption of the vote at the time of casting, and posting this to a public \emph{Bulletin Board} (\emph{BB}). Voters can then confirm that their ``receipt,'' i.e., the encrypted vote, appears correctly on the \emph{BB}.

Once we have consensus on the correct set of encrypted votes, these can be processed in a verifiable fashion to calculate the outcome in a way that does not compromise the privacy of the votes.
For instance, the encrypted votes might be put through a sequence of verifiable re-encryption mixes and then verifiably decrypted, allowing anyone to compute the result.
Alternatively, the encrypted votes might be tallied under encryption, exploiting the homomorphic properties of the encryption algorithm, and the final result decrypted.
To complete the assurance argument we need some additional ingredients:

\begin{itemize}
    \item The voter needs to be confident that her intended vote is correctly encrypted in her receipt.
    \item We need to prevent ballot stuffing, i.e., we need to ensure that only legitimately cast votes appear in the list of receipts on the \emph{BB}, and only one per voter.
    \item We need to know that ``enough'' voters check that their intended votes are correctly encrypted, and that their encrypted votes appear on \emph{BB}.
    \item We need dispute-resolution mechanisms in place to ensure that if voters detect (or claim to detect) problems, the culprit can be identified and appropriate action taken.
\end{itemize}

Typically the first point is addressed by some form of cut-and-chose protocol, e.g. \emph{Benaloh Challenge}~\cite{Benaloh2006}, or a more sophisticated approach such as Neff's \emph{MarkPledge} scheme~\cite{neff}.
Ballot stuffing is usually countered either by procedural measures in the polling station, or by requiring that receipts be digitally signed by the voters. The former does not provide \emph{universal eligibility verifiability} while the latter can but requires infrastructure to equip voters with signing keys.


We will not delve deeper into how the various E2E~V schemes work but rather assume that the correct set of encrypted votes is posted to the \emph{Bulletin Board}.

\subsection{Ballot Privacy, Receipt-Freeness and Coercion Resistance}

Ballot privacy is often defined using anonymity style definitions as originally proposed in \cite{DBLP:conf/esorics/SchneiderS96}. Informally, consider two instances of the system, one in which $A$ votes for $X$ and $B$ for $Y$, and the other in which the votes are swapped. If the attacker is unable to distinguish these two instances then the system is deemed to satisfy \emph{ballot privacy}.
More formal definitions can be found, for example, in~\cite{Delaune:2010}. Note that even in extreme cases, for example when all voters vote for $X$, such a system will satisfy the above definition, even though in that case the attacker knows precisely how each voter voted.

It was later realised that simple notions of ballot privacy in the presence of a passive attacker are not enough. For E2E V schemes we have to worry about ways that the voter might be able to prove her vote to a third party. This motivates the requirement for \emph{receipt-freeness}~\cite{Benaloh94receipt}: the voter cannot acquire evidence that would enable her to construct a proof to a third party as to how she voted.

In the face of a yet more active attacker who might interact with the voter before, during and after voting, potentially issuing detailed instructions and requiring the voter to reveal credentials, ephemeral random values etc, we need even stronger notions. This threat model motivates the property of \emph{coercion resistance} for which many different definitions have been proposed~\cite{juels05coercion}, reflecting various subtle distinctions. We will adopt the following definition, informally stated:
\begin{quote}
A voting system $S$ is coercion resistant if,
for all $c \in \mathcal{C}$ there exists a voter strategy $\psi$ such that for all attacker strategies $\phi$,  the voter can cast her intended vote $c$ and the attacker cannot tell that she did not obey his instructions.
\end{quote}



Such a style of definition appears to be the most powerful in that it captures the privacy failure in the case of unanimous votes, forced abstention and randomisation attacks.



\subsection{Selene}

We now give a sketch of how voter-verification is achieved in the Selene voting protocol. Full details can be found in \cite{selene}. In Selene, the verification is much more direct and intuitive than is the case for conventional E2E~V systems: rather than checking for the presence of her encrypted vote on the \emph{BB}, the voter checks her vote in cleartext in the tally on the \emph{BB} identified by a secret, deniable tracker.

During the setup phase the set of distinct trackers are posted on the $BB$, verifiably encrypted and mixed and then assigned to the voters according the resulting secret permutation. This ensures that each voter is assigned a unique, secret tracker.

For each encrypted tracker, a trapdoor commitment is created for which the voter holds the secret trapdoor key. In essence this is the ``$\beta$'' term of an El Gamal encryption of the tracker, where the ``$\alpha$'' term is kept secret for the moment.

Voting is as usual: an encryption of the vote is created, and sent to the server for posting to the $BB$ against the voter (pseudo)Id. Once we are happy that we have the correct set of validly cast, encrypted votes, we can proceed to tabulation: the (encrypted vote, tracker) pairs are put through verifiable, parallel re-encryption mixes and decrypted, revealing the vote/tracker pairs in plaintext.

Later, the $\alpha$ terms are sent via an untappable channel to the voters to enable them to open the commitment using their secret, trapdoor key. If coerced, the voter can generate a fake $\alpha$ that will open her commitment to an alternative tracker pointing to the coercer's choice. With the trapdoor, creating such a fake $\alpha$ is computationally straightforward. On the other hand, computing a fake $\alpha$ that will open the commitment to a given, valid tracker is intractable without the trapdoor. Thus, assuming that the voter's trapdoor is not compromised, the $\alpha$ term is implicitly authenticated by the fact that it opens to a valid tracker.

\section{Risk-Limiting Audits}\label{sec:RLA}

A \emph{risk-limiting audit} (RLA)~\cite{lindemanStark12} of a reported election outcome
is any procedure that has a known minimum chance of correcting the reported outcome if the outcome is wrong
(and that cannot render a correct outcome incorrect).
In this case, the \emph{outcome} means the winner or winners, not the precise tally.
The reported outcome is \emph{correct} if it is the outcome that an accurate
manual tally of the underlying voter-verified records would show.\footnote{%
The trustworthiness of the underlying records should be assessed
by a \emph{compliance audit}~\cite{starkWagner12}.
A RLA that relies on an untrustworthy record cannot reliably assess
whether outcomes reflect how voters voted.
}
The maximal chance that the procedure will fail to correct an outcome
that is wrong is the \emph{risk limit}.

RLAs generally pose auditing as a sequential test of the hypothesis that
the reported outcome is incorrect.
The audit continues to examine more ballots until either the hypothesis
is rejected or the audit has conducted a full manual tally.
The use of sequential tests enables RLAs to stop as soon
as there is convincing evidence that the reported outcome is correct, reducing the number of ballots the audit inspects.

RLAs check reported outcomes while RLTs determine what outcomes to report. However, similar sequential testing methods can allow RLTs to stop the tally
(of a random permutation of the ballots) as soon as there is convincing statistical evidence
of the electoral outcome, which the RLT then reports.
A RLT declares ``either this is the correct outcome, or an event occurred that had probability no larger than $\alpha$,'' where $\alpha \in (0, 1)$
is any pre-specified risk limit.
Minimizing the number of ballots that must be tallied maximizes the number of ballots kept shrouded, improving privacy and coercion-resistance.

There are two general strategies for RLAs: \emph{ballot-polling} and \emph{comparison}.
Ballot-polling manually examines randomly selected ballots for evidence of who won.
A comparison audit has three steps: first, the voting system must commit to its interpretation of physically identifiable individual ballots or groups of ballots comprising all ballots validly cast in the election.
Second, auditors check that the exported data reproduces the reported results.
Third, auditors compare the manual interpretation of a random sample of ballots or groups of ballots to the voting system's interpretation.
Further, ``hybrid'' methods combine ballot-polling for some groups of ballots and comparisons for other groups; see \cite{ottoboniEtal18}.

Comparison audits require auditors to know how the equipment interpreted the ballots, so they are not suitable for RLTs, where we seek evidence
about who won just from a subset of the shrouded votes.
Below, we show how a new procedure for ballot-polling RLAs can be adapted for RLTs.

\section{Risk-Limiting Tallies}\label{sec:RLT}

We propose a simple modification of the way that votes are tallied to address the issues outlined in the introduction.
Rather than tallying all votes straight-off, the election authority reveals the votes for a random sequence of encrypted ballots, continuing until the sample gives the acceptable level of risk in the outcome (i.e., who won).
If the true margin of victory is not too small, the outcome can be determined with high confidence (i.e., low risk)
while leaving a substantial number of ballots unopened, thus allowing a voter to claim that they cast the ballot required
by the coercer even if such a ballot was not revealed during the partial tally. 
The approach is thus inspired by the idea of Risk-Limiting Audits (RLAs), \cite{stark08a,lindemanStark12},
but here we apply the approach to  determining the correct outcome rather than checking whether a reported outcome is correct.
That difference turns out to have surprising statistical implications;
in particular, larger sample sizes are generally required to control the risk to the same level.

RLAs test the null hypothesis that the reported winner(s) did not actually win, rather than determine the correct outcome
\emph{ab initio}.
Moreover, the operating characteristics of existing RLAs depend on the reported results.
For instance, \emph{comparison audits} test whether the reported margin overstated
the true margin by enough to cause the reported winners to be incorrect.
Previous methods for \emph{ballot-polling} audits, such as BRAVO \cite{lindemanEtal12} test the hypothesis that the reported outcome is wrong against the alternative
that the reported vote shares are nearly correct.

For RLTs, we do not have reported results to leverage, so we need a new approach.
Section~\ref{sec:mean} presents a probability inequality; Section~\ref{sec:RLTallies} applies it to
produce a new sequential ballot-polling test, the engine for the RLT scheme presented in Section~\ref{sec:protocol} based on the Selene E2E~V protocol.

\subsection{Tests for the Mean of a Non-Negative Population}\label{sec:mean}

Extant methods for RLAs generally involve the reported results in
some way.
Here, we present a new sequential method to determine with high
confidence who won, without specifying a particular alternative hypothesis.
The method applies to plurality (including vote-for-$k$), majority, and super-majority social choice functions,
but we present the method in detail only for plurality contests.

Our RLT method is based on tests about the mean of a non-negative population.
Consider a population of $N$ items, each labeled with a non-negative number.\footnote{%
  In our case, the items will be ballots, and their labels will represent votes; see Section~\ref{sec:RLTallies}. }
Let $x_i \ge 0$ be the label of item $i$, $i=1, \ldots, N$.
Let $\mu \equiv \frac{1}{N}\sum_{i=1}^N x_i$ be the mean
of the labels.
Moreover, let $t$ denote the hypothesized value of the population mean $\mu$.

We sample items at random, sequentially, without replacement,
such that the (conditional) probability that item $k$ is selected in the $j$th draw is $\frac{1}{N-j+1}$,
given that item $k$ was not selected before the $j$th draw. $X_j$ denotes the number on the label of the item selected on the $j$th draw.
Define $S_j \equiv \sum_{k=1}^j X_k$, $\tilde{S}_j \equiv S_j/N$, and $\tilde{j} \equiv 1 - (j-1)/N$.
Let
\begin{equation} \label{eq:martDef}
   Y_n \equiv \int_0^1 \prod_{j=1}^n \left ( \gamma \left [ X_j \frac{\tilde{j}}{t - \tilde{S}_{j-1}} -1 \right ] + 1 \right ) d\gamma.
\end{equation}
It has been shown in~\cite{evansStark19} that if $\mu=t$ (i.e., if the null hypothesis is true), then $(Y_j )_{j=1}^N$ is a nonnegative closed martingale with expected value 1.
Kolmogorov's inequality then implies that for any $J \in \{1, \ldots, N\}$ and any $p \in (0, 1)$,
$$
   \Pr \left ( \max_{1 \le j \le J} Y_j(t) > 1/p \right ) \le p.
$$
This can be used as the basis of a ballot-polling RLA that does not require a reference tally, as we show below.
The same result holds for sequential sampling \emph{with} replacement, re-defining
$\tilde{S}_j \equiv 0$ and $\tilde{j} \equiv 1$
(the limit of the finite-population result as $N \rightarrow \infty$).
We also note that \cite{evansStark19} provides a recursive algorithm for computing the integral \eqref{eq:martDef}.

\subsection{Risk-Limiting Tallies}\label{sec:RLTallies}

Consider plurality contests that allow each voter to vote for $k \ge 1$
of $C$ candidates.
The winner(s) are the $k$ candidates who receive the most votes.
We ignore the possibility of ties; they are an easy extension.
Majority and super-majority are straightforward generalizations; see \cite{stark08a}.

Candidate $w$ is one of the winners if $w$ received more votes than at least $C-k$
other candidates.
In general, some ballots will have invalid votes or votes for other candidates.
Consider a single pair of candidates, $w$ and $\ell$.
Let $N_w$ denote the number of ballots in the population that show a vote for $w$ but not
for $\ell$; let $N_\ell$ denote the number of ballots in the population that show a vote for $\ell$ but not for $w$,
and let $N_u \equiv N - N_w - N_\ell$ denote the number of ballots
that show a vote for neither $w$ nor $\ell$ or show votes for both $w$ and $\ell$.

Let $W_j$ be the number of items labeled with $w$ selected on or before draw $j$; and define $L_j$ analogously.
The probability distributions of those variables 
depend on $N_w$, $N_\ell$, and
$N_u$, even though we only care about one parameter, $N_w - N_\ell$.
Now $N_w \le N_\ell$ if and only if $N_w + N_u/2 \le N_\ell + N_u/2$.
Since $N_\ell + N_u/2 = N - (N_w + N_u/2)$,
we have $N_w + N_u/2  \le N - (N_w + N_u/2)$.
We can now divide by $N$ to obtain $\frac{N_w + N_u/2}{N}  \le  1 - \frac{N_w + N_u/2}{N}$ from which we get
\begin{equation}
    \frac{N_w + N_u/2}{N}  \le \frac{1}{2}\ .
\end{equation}


Let
$$
   \mu_{w\ell} \equiv \frac{1 \times N_w + \frac{1}{2} \times N_u + 0 \times N_\ell}{N}\ .
$$
This is the mean of a population derived from re-labeling each vote for $w$ as 1,
each vote for $\ell$ as 0, and the rest as $1/2$.
The mean of this population is greater than $1/2$ iff $w$ received more votes than $\ell$.
We can test the hypothesis $\mu_{w\ell} \le 1/2$ (i.e., $w$ did not beat $\ell$)
using the martingale-based test above by simply treating the sampled ballots that way: every ballot with a vote for $w$ (but not $\ell$) counts as $1$, every ballot with a vote for $\ell$ (but not for $w$) counts as $0$, and invalid ballots, ballots with votes for other candidates, and ballots with votes for both $w$ and $\ell$ count as $1/2$.

To determine the set of winners, we sequentially test the collection of
$C(C-1)$ hypotheses
\begin{equation}
   \{H_{w\ell}: \mu_{w\ell} \le 1/2, w = 1, \ldots, C; \ell = 1, \ldots, C; w \ne \ell \},
\end{equation}
stopping when either
\begin{itemize}
    \item there is a set $\mathcal{W}$ of cardinality $k$ such that we have rejected the hypothesis $\mu_{w\ell} \le 1/2$ for every $(w, \ell)$ with $w \in \mathcal{W}$ and $\ell \notin \mathcal{W}$, or
    \item we have examined a too high percentage of votes from the privacy point of view, in which case the sampling strategy is abandoned and different means are used to determine with certainty who won, see Section~\ref{sec:fallback} for details.
\end{itemize}

\begin{proposition}
  If every hypothesis is tested at level $\alpha$,
  the probability that this algorithm misidentifies the set of winner(s) is at most $k(C-k)\alpha$.
\end{proposition}

\begin{proof}
The approach misidentifies one or more winners iff it terminates in the first branch, but $\mathcal{W}$ is not the set of winners:
$\exists w \in \mathcal{W}$, $\ell \notin \mathcal{W}$ s.t.
$\mu_{w\ell} \le 1/2$.
In a RLA, a wrong outcome can only be confirmed if \emph{every} true null hypothesis is erroneously rejected.
In contrast, in a RLT, a wrong outcome can be confirmed if just one particular true
null hypothesis is rejected: the hypothesis that the candidate with the $k+1$st
highest vote share got fewer votes than the candidate with the $k$th highest vote share.

There are $C(C-1)$ hypotheses $\{H_{w,\ell}\}$ in all, of which $C(C-1)/2$ are true.
Of the true null hypotheses, those whose erroneous rejection would make the reported outcome wrong
are the $k(C-k)$ that compare the vote share of a candidate in $\mathcal{W}$ to the vote share of a candidate in $\mathcal{W}^c$:
if none of those is erroneously rejected, the set of winners is correct.
Observe that if we used the logical implications of the statistical rejections to entail rejections of other
hypotheses---for instance, $H_{w\ell} \cap H_{\ell k} \rightarrow H_{w k}$---this would not be true.
Therefore, a Bonferroni multiplicity adjustment of $k(C-k)$ certainly suffices.
Note that this may be conservative as an estimate, because there are logical dependencies among the hypotheses.
\qed
\end{proof}
The aim of the sampling is to test the hypothesis ``$\mu_{w\ell}\le 1/2$.'' Rejecting $\mu_{w\ell}\le 1/2$ means proving with
risk at most $\alpha$ that $w$ won the pairwise contest with $\ell$.

\begin{proposition}
If we reject $\mu_{w\ell} \le 1/2$ at significance level $\alpha$ and reject $\mu_{\ell m} \le 1/2$ at significance level $\alpha$, then we
reject $\mu_{wm} \le  1/2$ at significance level $\alpha$.
\end{proposition}

\begin{proof}[sketch]
This transitivity property follows from the monotonicity of the $P$-values in the number of votes for each candidate, at each sample size $j$.
\qed
\end{proof}

\subsection{Sample sizes}
Because the underlying statistical test is sequential, the audit can start by looking at a single ballot selected at random,
calculate the $p$-values for all not-yet-rejected null hypotheses,
and continue to increase the sample one ballot at a time
until the risk limit has been met.
However, depending on the desired risk limit, the RLT will not be able to terminate until some
minimum number of ballots has been tallied.

The minimum sample sizes required to identify the winner with a maximum error rate of $\alpha$ are given in Table~\ref{tab:sample_sizes}, for sampling
without replacement, for a plurality contest with 2~candidates and a plurality contest with 10~candidates.
The sample sizes listed are exactly those that would be required
if the votes were unanimously for one candidate; 
if more than one candidate receives votes, the sample size becomes random and becomes stochastically larger.

Similarly, if a fraction $u$ of ballots do not have a valid vote for any candidate,
the sample size will also be random, and the expected sample size will grow by a factor of $1/(1-u)$.
For instance, if 10\% of ballots
have no vote and 90\% of ballots have a vote for candidate~A in 
a 10-candidate plurality election, the expected sample size to identify
the winner with risk limit $0.1\%$ is $17/0.9 = 18.9$ ballots.

The closer the vote is to unanimous, the fewer ballots need to be revealed (the distribution is stochastically smaller the more nearly unanimous the vote).
I.e., the protection a RLT offers is greatest when the risk is greatest.

For a two-candidate plurality election, only one of the two null hypotheses $\mu_{\ell m} \le 1/2$
can be true; thus, no multiplicity adjustment is needed.
(This is consistent with the formula $k(C-k) = 1\times(2-1) = 1$.)
For a 10-candidate plurality election, the Bonferroni adjustment factor is $1\times(10-1) = 9$.
As the table shows, if the vote is (nearly) unanimous, the number of ballots required to identify the winner
with negligible error probability is small: 35 suffices to have an error probability less than $10^{-9}$ for a two-candidate contest, and 38 suffices for a 10-candidate contest.
Because the risk drops by an order of magnitude with an increase in sample size of about 4 ballots when the vote is (nearly) unanimous, the penalty for multiplicity is low in absolute terms. If the RLT sample is drawn \emph{without} replacement, the expected sample sizes required to attain a given risk are smaller---but not by much unless the total number of ballots 
is small.

\begin{table}[]
    \centering
    \begin{tabular}{c|ccccccccc}
    candidates  & \multicolumn{9}{c}{ $\alpha$ } \\ 
    & $10^{-1}$\, & $10^{-2}$\, & $10^{-3}$\, & $10^{-4}$\, & $10^{-5}$\, & $10^{-6}$\, & $10^{-7}$\, & $10^{-8}$\, & $10^{-9}$\, \\ \hline
    2 &  5  &  9&  13& 17& 21& 24& 28& 31 & 35 \\
    10 & 9& 13& 17& 20 & 24& 27& 31& 34 &  38
    
    %
    \end{tabular}
    \caption{Minimum sample sizes to identify the winner of a two-candidate plurality contest and a 10-candidate plurality contest at risk limit $\alpha$, for sampling with replacement. Actual sample sizes approach these minima (with high probability) as voter preferences approach unanimity.}
    \label{tab:sample_sizes}
\end{table}

\section{Incorporating RLT in E2E~V Voting Protocols}\label{sec:protocol}
RLTs can be used in a straightforward way with any E2E~V scheme in which the set of encrypted votes appears on a Bulletin Board (\emph{BB}) and is applicable to either remote or in-person voting.
The encryption should be homomorphic and probabilistic: for instance, ElGamal can be used. Helios~\cite{ampq:helios}, Pr\^et \`a Voter~\cite{crs:pav}, Selene~\cite{selene}, etc., would all be amenable. 

Conceptually, we can start with a random permutations of the encrypted votes and take samples from left to right, opening more ballots as required. The verifiable shuffles used in many schemes naturally give us a random permutation.
However, we must be careful about simply taking the permutation output of the underlying scheme's shuffles, as there may be opportunities
to manipulate this and bias the sampling. 
%
%
The sampling must be truly random and demonstrably outwith the control of any entity. This brings us to the challenge of \emph{certifiable randomness}, which arises in many contexts: lotteries, voting, auctions, public ledgers etc. A number of approaches have been proposed, for example using a seed derived from a hash of prices of previously agreed stock market options at an agreed future time.  Alternative approaches involve combining  random values previously committed by a number of independent entities. Algorand~\cite{DBLP:journals/corr/Micali16} adopts such an approach combined with the use of verifiable random functions. Another possibility is to derive the seed from a cryptographic hash of suitable data posted to the $BB$. RLAs have employed seeds generated in a public ceremony of dice rolling. We might rely on a trusted third party such as the NIST random beacon service. 
For the purposes of this paper do not specify a particular approach but leave it for the stakeholders to select.

Sampling with replacement can be implemented straightforwardly by performing further mixes between samplings.

\subsection{Guaranteeing Plausible Deniability}\label{sec:fallback}

For most elections, the RLT approach will naturally leave a good proportion of unrevealed votes. However, there will be cases where the winning margins are narrow, and thus the RLT might result in all or almost all votes being revealed. 
It is not enough for a system to be (objectively) coercion resistant, it must also be seen as coercion resistant.
Thus, for the RLT approach to be effective, we must ensure that the voters will never be, nor expect to be, in a situation in which plausible deniability fails. In this section we identify such situations, and describe some strategies to deal with the potential vulnerability.

Of course, a close run referendum will not be a problem, but a problematic scenario is a close margin between candidates $X$ and $Y$, along with a low-support candidate $Z$. This could result in a full count where the low score $Z$ opens up the possibility of coercion.
We have already indicated that coercion for $Z$ is possibly harmful for the outcome of the election, as it can be used to decrease the number of votes that either $X$ or $Y$ gets.
Note also, again, that this kind of coercion is feasible not only when the voter \emph{knows} (e.g., from polls) the a close run will occur. In many cases, it suffices that the voter thinks it \emph{might} happen to get her worried and vulnerable to threats.
We propose that, in such circumstances, the system should switch to a \emph{fallback strategy} that works in all cases. Example fallback strategies are sketched below.


\para{PET testing.}
In the event of a close race between $X$ and $Y$, start Plaintext Equivalence Testing of randomly selected, unrevealed ballots against $\{X\}_{PK}$ and $\{Y\}_{PK}$, until we reach the required confidence for the winner.



\para{Tally hiding.} One can also fall back to computationally heavy methods e.g. MPC for only disclosing the winner, see e.g. \cite{Cohen86improvingprivacy,193463,10.1007/978-3-319-98989-1_17,teague}. Note that the revealed votes and reduced number of possible winners will make these methods more efficient than if used from the onset.

A possibility is to have the tellers perform a secret computation of the tally, and announce the winner(s), but not the numerical tally, on which to base a null hypothesis. This allows the RLT to be computed much more efficiently, and the secretly computed tallies can guide the appropriate strategy to adopt in the event of narrow margins.

\section{Security Assumptions}\label{sec:security}

In this section we briefly state the security guarantees and give some arguments for their validity. For the exposition below, we introduce the following three authorities besides the voters: The Tally Tellers $\TT$ holding the secret election key in a threshold manner, the Mixnet Tellers $\MT$ mixing the encrypted votes before doing the risk-limiting tally, and a random sampling authority $\RSA$ organising the random sampling of votes for the tally.

For simplicity let us also assume that the underlying voting scheme that we build on is mixnet-based, i.e., the main difference between the RLT version and the original version is that not all ciphertexts output from the mixnet are decrypted, but only a proportion of them.

In general, if $\RSA$ is acting honestly, or bound to do so e.g. via a verifiable proof based on a computational assumption, then the security reduces to that of the underlying scheme. For privacy, we normally have to trust that a threshold set of $\TT$ is not colluding and at least one server in $\MT$ is honest. For verifiability most schemes will not impose verifiable trust in $\TT$ or $\MT$ but might rely on computational assumptions and the RO-model or a CRS setup.

\para{Verifiability. } When random sampling procedure is corrupted, the adversary could possibly adjust the outcome in his favour. However, note that in this situation we can still achieve verifiability by having $\RSA$ committing to the sampling order before mixing, and assuming the last mix node is honest (or assuming one arbitrary mix node is honest and no threshold set of $\TT$ is corrupted). This will ensure that the final sampling is random.

\para{Ballot-Privacy.} Obviously, a necessary assumption for ballot-privacy is that a threshold set of $\TT$ are not colluding, and at least one mix node is honest. If the random sampling is also honest, we get strictly less information from the tally than in the original scheme, and we thus achieve better privacy in an information theoretic sense.
When using standard ballot-privacy definitions on the scheme it should also be possible to reduce the ballot-privacy to that of the underlying scheme, the only subtlety being that tally functions differ in the two schemes.

It might seem that a corruption of the random sampling procedure should not influence ballot-privacy. However, there is one assumption to make: the random sampling should, in the computational view of the adversary, be uncorrelated with the cast votes. Having input from the voters to the random sampling could indeed make sense from a verifiability viewpoint, like in Demos \cite{DBLP:conf/ccs/Kiayias0Z15}, but should not depend on the vote choice unless this is computationally hidden.

\para{Coercion-Resistance and Vote-Buying Resistance.} As we have discussed above, the RLT protocol in general improves the coercion-resistance especially when candidates are expected to have a low vote count. It would be interesting to relate this to the coercion-resistance level $\delta$ in Kusters et al. \cite{DBLP:conf/csfw/KustersTV10}. On the other hand, the security against vote-buying is not increased in the same way since the voter here has an intent to obtain a receipt. The vote buyer could indeed follow the Italian attack method and the marked ballot would often appear.

The above is reminiscent of a distinguishing example between vote-buying and coercion resistance due to Rivest:\footnote{private communication} the system chooses at random whether or not to provide the voter with a plaintext receipt. Such a system is, arguably, coercion resistant (the voter can claim to have received no receipt) but is vulnerable to vote buying (the voter might comply in the hope of getting the pay-off).

\section{Risk-Limiting Verification}\label{sec:rlv}

The idea of using risk-limiting techniques to improve coercion resistance can also be applied to verification of votes. Here, we apply the idea to Selene, allowing us to ensure that a proportion of the trackers remain unrevealed. In consequence, the coerced voter can always claim that her tracker was amongst those that remained shrouded. Some subtleties have to be handled in the case of an obnoxious coercer who demands the voter divulge their tracker; we describe those below. Indeed, these consideration require some modifications of the way Selene works.

\subsection{Risk-Limiting Verification in Selene}

A drawback of Selene, as noted in the original paper, is that when a coerced voter claims a fake tracker, the coercer (who is also a voter) could maintain that this is in fact his tracker. By construction, the coercer cannot prove this to the voter, but the voter is now in a difficult position: she knows that the claim might be true. Elaborations of the basic scheme are proposed, but they complicate things and render the verification less transparent: the final tally contains dummy votes that must be subtracted out to get the true result.

The RLT idea can be extended to avoiding revealing all the trackers in a run of a Selene election. The natural step is to apply the RLT mechanisms described above to reveal as many votes as necessary and then reveal the corresponding trackers. There would seem to be little point in revealing trackers for which the corresponding vote has not been revealed. There may, however, be some merit in revealing a subset of the trackers for which the votes have been revealed, as we discuss below.

Risk Limiting Verification (RLV), as applied to Selene, can ensure that not all trackers are revealed, thus allowing a coerced voter to simply claim that their tracker did not appear. There is still a problem, however, if we use Selene in its original form: the full set of trackers is published, so the coercer could require the voter to reveal her tracker anyway, and still claim that it is his.

We can fix this fairly easily: one purpose of revealing all the trackers in the setup phase is to demonstrate that they are all distinct, so the EA could publish a list of encrypted trackers for which the trustees run pairwise PETs to show that all the plaintexts differ. This would be computationally heavy and does not scale well, but is no worse than, e.g., JCJ~\cite{juels05coercion}. Moreover, we can use some of the approaches to linearising the JCJ-style checks, for example by raising the tracker ciphertexts to the same, secret exponent and then verifiably decrypting.

We note that we get a form of partial random checking anyway when we reveal a random sample of the trackers: if all the revealed trackers in say a 90\% random sample are distinct then we have very high confidence that they all are distinct. The drawback of this approach is that if the EA has cheated and included collisions then we will not discover this until rather late. Note, however, that we could reveal the trackers first, before revealing the votes. Now, if we find collisions, we can abort the election before any tally results have been revealed.

Still, one problem remains: another reason to publish the set of trackers is to allow voters to confirm that the $\alpha$ term sent to them is authentic: it opens the commitment to a valid  tracker, i.e., a member the published set. If we do not publish the set of trackers then we need another mechanism for voters to confirm that their notified tracker is ``valid.'' We can  achieve this by requiring that valid trackers are drawn from a negligible subset of the full space, e.g., numbers with say six digits. Now it is still intractable to produce fake $\alpha$ terms that will open a given commitment  to a member of this set, but, by adjusting the number of digits we can ensure that the chance that a fake tracker will collide with the coercer's is greatly reduced, so improving the plausible deniability.


If this reduced probability of tracker collision is deemed unacceptable, then we could allow the voter to request a fake tracker from the Notification Authority. This authority knows which valid trackers have not been assigned and so can provide an unassigned tracker to the coerced voter. This requires a level of trust in this entity, to keep tracker-related information secret but such trust is needed anyway.

There remains the question of whether all the voters should be notified of their tracker, even when their tracker has not been revealed on the BB.  The immediate thought is not to notify unrevealed trackers, but this introduces possibilities of the authorities exploiting this: leading many voters to think that their tracker was not revealed and so denying them the possibility to verify their vote. It is not clear how we could verify that all the voters whose trackers are revealed are notified, so it seems wiser to notify each voter of their tracker. 




\section{Related Work}\label{sec:relatedwork}

A number of papers \cite{Cohen86improvingprivacy,193463,10.1007/978-3-319-98989-1_17,teague} try to achieve tally hiding, either by only calculating the winner(s), or via multi-party computation and other cryptographic means.
An idea closer to RLTs is that of Random Sample Voting (RSV) by Chaum~\cite{chaumRSV}.
A scheme that seeks to implement RSV in a fully verifiable fashion is Alethea~\cite{alethea}.
RSV typically samples a small and predetermined number of voters, regardless of the margins.
In contrast, RLTs adjust the sample size to obtain the desired level of confidence in the reported outcome.


The idea of Risk-Limiting Verification is somewhat analogous to Rivest's ThreeBallot protocol~\cite{threeballot}.
Recall that, in ThreeBallot, each voter can verify a random 1/3 of her cast ballot.
Thus, RLV gives ``vote handles'' to a fraction of voters, whereas in ThreeBallot each voter gets a handle to a fraction of her vote.

\section{Conclusions}\label{sec:conclusions}

This paper presents two simple methods, RLT and RLV, for reducing the amount of information provided in the tally and verification stages. In consequence, we enhance the coercion-resistance by giving coerced voters plausible deniability, while achieving whatever confidence level in the outcome is required. An important future step will be to understand how well this method protects coerced voters in practice. It would be good also to understand better the trade-off between confidence in the outcome and plausible deniability levels.

There exist other methods that leak less information in the tally process, e.g., by using multi-party computation to only reveal the winner of the election. Such methods might be better suited to avoid strategic voting in runoff elections, and may provide somewhat better deniability.
However, those methods require more elaborate and computationally expensive cryptography; arguably, our methods are more efficient and transparent.

The novel Risk-Limiting techniques introduced here should be of independent interest and have applications beyond the RLTs and RLVs described here.





\para{Acknowledgements.}
WJ and PYAR acknowledge the support of the Luxembourg National Research Fund (FNR) and the National Centre for Research and Development (NCBiR Po\-land) under the INTER/PolLux project VoteVerif (POL\-LUX\--IV/1/2016). PBR was supported by the EU Horizon 2020 research and innovation programme under grant agreement No. 779391 (FutureTPM).

%

\end{document}